\documentclass[sigconf]{aamas} 

\usepackage{balance} % for balancing columns on the final page
\usepackage{amsthm,mathtools}
\usepackage{dsfont,bm}
\usepackage{units}
\usepackage{multirow,array}
\usepackage[ruled]{algorithm2e}
\newtheorem{prop}{Proposition}

\setcopyright{ifaamas}
\acmConference[AAMAS '21]{Proc.\@ of the 20th International Conference on Autonomous Agents and Multiagent Systems (AAMAS 2021)}{May 3--7, 2021}{Online}{U.~Endriss, A.~Now\'{e}, F.~Dignum, A.~Lomuscio (eds.)}
\copyrightyear{2021}
\acmYear{2021}
\acmDOI{}
\acmPrice{}
\acmISBN{}

\acmSubmissionID{181}
\title[AAMAS-2021]{Accumulating Risk Capital Through Investing in Cooperation }

\author{Charlotte Roman, Michael Dennis, Andrew Critch, Stuart Russell}
 \affiliation{
  \department{Center for Human-Compatible AI }
  \institution{University of California, Berkeley}}
\email{c.d.roman@warwick.ac.uk, {michael_dennis,critch,russell}@cs.berkeley.edu}
\begin{abstract} 
Recent work on promoting cooperation in multi-agent learning has resulted in many methods which successfully promote cooperation at the cost of becoming more vulnerable to exploitation by malicious actors.  We show that this is an unavoidable trade-off and propose an objective which balances these concerns, promoting both safety and long-term cooperation. Moreover, the trade-off between safety and cooperation is not severe, and you can receive exponentially large returns through cooperation from a small amount of risk. We study both an exact solution method and propose a method for training policies that targets this objective, Accumulating Risk Capital Through Investing in Cooperation (ARCTIC), and evaluate them in iterated Prisoner's Dilemma and Stag Hunt.
\end{abstract}

\keywords{Social Dilemmas; Game Theory; Multi-agent Reinforcement Learning; Safety; Cooperation}

         \usepackage{todonotes}

\begin{document}
%% Page limit: 8 pages plus references  
%%% The following commands remove the headers in your paper. For final 
%%% papers, these will be inserted during the pagination process.
\pagestyle{fancy}
\fancyhead{}

\maketitle 
%%%%%%%%%%%%%%%%%%%%%%%%%%%%%%%%%%%%%%%%%%%%%%%%%%%%%%%%%%%%%%%%%%%%%%%%

\section{Introduction}
%Main Points:
% 1) Social Dilemmas are a core problem in MAL
% 2) Many existing methods which maximize social welfare do so at the expense of individual rationality and safety.
% 3) To some extent, these two ideas are fundamentally at odds, cooperation definitionally leaving room for exploitation.
% 4) We characterize the trade off between these two ideas,  and create an objective which optimally trades off between Safety and cooperativeness.  Importantly, we show that the trade off between safety and cooperativeness decreases exponentially in the length of the interaction.  Thus, in long iterated games, the optimal policy for this objective is nearly-rational for the designer to deploy into either fully adversarial settings, or fully cooperative settings.  
% 6) Finally we propose a method to train policies which satisfy this objective, which we call Safe Investment of Effective Capital (SIEC).  The key idea behind this method builds on previous work showing that $\epsilon$-safe policies risk what they've won in expectation.  We call this amount that they're allowed to risk, their "effective capital" and motivate the agents to take that risk in a cooperative way, with the belief that any investment of effective capital towards cooperation will be repaid with interest.

Sequential social dilemmas (SSDs) are games where short-term individual incentives conflict with the long-term social good.  Such games are pervasive and describe situations in which we would deploy automated multi-agent systems.  In many of these situations, we only control one of the agents involved, as is the case with self-driving cars where the other agent could be a human or a self-driving car from another company.  Many methods have been proposed for training policies that would better optimize social good~\cite{Lerer2017,Hughes2018, Eccles2019, Jaques2019}.  However, many of these methods do so at a greater risk of being exploited. %, which is especially problematic when you do not control the other agents.
A policy that optimizes for social welfare could allow other policies to be selfish without long-term consequences. So the welfare-maximizing policy may fare worse in these settings than
if it had only optimized for its own self-interest. To some extent, this is unavoidable since every choice to cooperate leaves room to be exploited.  However, though this trade-off is inevitable, it is not as stark as it first appears.

To study this trade-off formally, we present two objectives. The first objective is the well-studied notion of $\epsilon$-safety. Previous work shows that $\epsilon$-safe policies in sequential settings risk what the policy won in expectation \cite{Ganzfried2012}, allowing them in the long run to take much larger risks. The second objective is to perform well with cooperation-promoting policies. We define cooperation-promoting policies to be such that they are more cooperative when faced with cooperative policies. % The existence of cooperation-promoting policies is the reason why a rational designer might consider making an agent that cooperates.  
Since perfect safety would require only defection and perfect performance with cooperation-promoting policies would require only cooperation, these two objectives are in clear conflict.

Throughout this work, we find it convenient to conceive of $\epsilon$-safety using the concept of risk capital.  Risk capital is the amount of capital an investor is willing to lose.  In our case, we use risk capital to refer to the amount of utility we are willing to risk, which for $\epsilon$-safe policies is $\epsilon$.  In this framing, we can think of prior work, showing that $\epsilon$-safe policies risk what the policy won in expectation, as noticing that a policy does not need more risk capital for reinvesting their unexpected winnings.  
Since cooperation could always be met with defection, cooperating necessitates some amount of willingness to lose utility.  However, like an investment, cooperation often results in returning more utility than was risked.  This returned utility can then be reinvested without risking worse outcomes, leading to growing risk capital over time, thus more cooperation over time even for a small initial amount of risk capital.

This argument shows an interesting fact about the trade-off between safety and cooperation -- that though cooperation with total safety is impossible, giving away even a small amount of safety will lead to nearly optimal cooperation in the long-term.  We formalize this argument as a trade-off between safe beliefs and cooperation promoting beliefs.  We propose a method to train policies that satisfy this objective, which we call Accumulating Risk Capital Through Investing in Cooperation (ARCTIC), based on the idea of investing risk capital to achieve long-term cooperation. We test this method in the simple domains of prisoners dilemma and stag hunt. However, we expect investing in cooperation to be a broadly applicable principle for developing safe and cooperative agents.  To achieve full safety, we begin with fully safe beliefs and adapt the safety levels according to opponent play. A positive initial risk capital can be used with minimal effect on long-term safety while achieving cooperation more easily.

%%%%%%%%%%%%%%%%%%%%%%%%%%%%%%%%%%%%%%%%%%%%%%%%%%%%%%%%%%%%%%%%%%%%%%%%
\subsection{Literature Review}
 The problem of cooperation in sequential settings has long been studied in game theory.  One of the most famous strategies is Rapport's tit-for-tat \cite{Rapoport1965}, which achieved great success in Axelrod's tournament \cite{Axelrod1981}.  In part, Axelrod attributed it's success to its ability to promote cooperation and is not
 exploitable after the first move.

 Human experiments of sequential social dilemmas often find cooperation as the most popular strategy. \citeauthor{Gardner1984} \cite{Gardner1984} use bounded rationality to explain the high levels of cooperation in lab experiments where participants played a common pool resource game and could communicate with each other. Human cooperation seen in social dilemmas is often founded in indirect reciprocity \cite{Rand2013}. One such suggestion, translucent game theory \cite{Capraro2019}, explains cooperative behaviors through the idea that others' could detect an intended defection and punish it, which serves as the inspiration for our use of policy-conditioned beliefs.  These transparency considerations, such as the ability to view opponent facial expressions, change behaviors in social dilemmas \cite{Hoegen2017,deMelo2018}.

 Sequential social dilemmas exist in a variety of real-world settings, such as traffic networks  \cite{Klein2016,Roman2021}, and are essential to the development of cooperative artificial intelligence \cite{Dafoe2020}.
 The social dilemmas we consider are Prisoner's Dilemma and Stag Hunt. In Prisoner's Dilemma, Press and Dyson \cite{Press2012} showed that there exist evolutionary dominant strategies that can only be outperformed by players who have a theory of mind about their opponent; and, when combining strategies with a memory of one and theory of mind, stable strategies have been classified \cite{Glynatsi2020}. For Stag Hunt, the effects of network topology are important for the emergence of cooperation in games played on a network \cite{Segboreck2010}. 
 
 %Triver's direct reciprocity, 
   \begin{table*}[h]
     \caption{Social dilemma payoffs are of the form shown (left) where $R>P$, $R>S$, $2R>T+S$, and either $T>R$, or $P>S$. Examples of payoff matrices for two-player SSDs (payoffs normalized on $[0,1]$): Stag Hunt (centre) and Prisoner's Dilemma (right).}
    \label{table:ssd}
    \vspace{-2mm}
      \begin{minipage}{0.3\linewidth}
    \centering
      \begin{tabular}{cc|c|c|}
          & \multicolumn{1}{c}{} & \multicolumn{1}{c}{$C$}  & \multicolumn{1}{c}{$D$} \\\cline{3-4}
           & $C$ & $(R,R)$ & $(S,T)$ \\\cline{3-4}
          & $D$ & $(T,S)$ & $(P,P)$ \\\cline{3-4}
          \\
        \end{tabular}
    \end{minipage}
    \begin{minipage}{0.3\linewidth}
    \centering
        \begin{tabular}{cc|c|c|}
         & \multicolumn{1}{c}{} & \multicolumn{1}{c}{$C$}  & \multicolumn{1}{c}{$D$} \\ \cline{3-4}
         & $C$ & $(1,1)$ & $(0,\nicefrac{3}{4})$ \\ \cline{3-4}
         & $D$ & $(\nicefrac{3}{4},0)$ & $(\nicefrac{1}{4},\nicefrac{1}{4})$ \\ \cline{3-4}
          \\
        \end{tabular}
    \end{minipage}
    \begin{minipage}{0.3\linewidth}
    \centering
        \begin{tabular}{cc|c|c|}
          & \multicolumn{1}{c}{} & \multicolumn{1}{c}{$C$}  & \multicolumn{1}{c}{$D$} \\\cline{3-4}
          & $C$ & $(\nicefrac{3}{4},\nicefrac{3}{4})$ & $(0,1)$ \\\cline{3-4}
          & $D$ & $(1,0)$ & $(\nicefrac{1}{4},\nicefrac{1}{4})$ \\\cline{3-4}
           \\
        \end{tabular}
    \end{minipage}
    \vspace{-2mm}
    \end{table*}
 In the reinforcement learning (RL) literature, temporal difference learning has been applied successfully to learn cooperation in Prisoner's Dilemma since it maximizes future rewards \cite{Masuda2009}. Leibo et al. \cite{Leibo2017} extended classic sequential social dilemma examples into the domain of deep reinforcement learning. Following on from this work, \citeauthor{Jaques2019} \cite{Jaques2019} used social influence as intrinsic motivation to achieve coordination and communication between agents in SSDs.
 Additionally, approximate Markov tit-for-tat \cite{Lerer2017} maintains desirable properties from tit-for-tat but applied to deep learning in general games. \citeauthor{Eccles2019} achieved reciprocity through a system of innovators who maximize their own rewards and imitators who learn to mimic their behaviors by measuring the niceness of actions \cite{Eccles2019,Eccles2019b}, resulting in good performance in SSDs. %limitation is it's not applicable to asymmetric games, niceness function has to be chosen for each environment
 
 There has also been work in the RL literature on promoting cooperation through opponent modeling. Learning with Opponent Learning Awareness (LOLA) exemplifies the work in this direction, which models the opponent's learning process to promote long-term cooperation \cite{Foerster2018}. Other work in this direction, such as  Convergence with Model Learning and Safety (CMLeS), models the memory bounded aspects of other agents to achieve cooperation \cite{Chakraborty2014}. Although all of these works promote methods of cooperation, they do so by changing the context to make cooperation an equilibrium or a stable point of the training regime. Instead of changing the context, we satisfy ourselves with policies that are only near equilibrium, in that they are $\epsilon$-safe, and construct a strategy that can cooperate in this context.
 
  Safe strategies are essential to restrict the impact of adversarial opponents \cite{Gleave2019}. \citeauthor{Ganzfried2012} \cite{Ganzfried2012} explore the necessary properties for a policy to allow for safe opponent exploitation. Yet, there is a gap in the literature for safety in SSDs. Most research is concerned with adapting tit-for-tat like mechanisms to handle adversaries, whereas our proposed algorithm guarantees the safety property.

\subsection{Paper Outline}
We begin in Section \ref{sec:prelims} by outlining the necessary mathematical preliminaries. In Section \ref{sec:safebeliefs}, we establish the conditions for the $\epsilon$-safety of beliefs. Then we define a policy-conditioned belief that incentivizes cooperation in SSDs in Section \ref{sec:coopbeliefs}. In Section \ref{sec:coopvsafe}, where we show that the trade-off between safety and cooperativeness decreases with length of the interaction. In Section \ref{sec:arctic}, we equip this cooperation inducing belief with the safety property and propose the ARCTIC algorithm. Additionally, we show ARCTIC has desirable properties in matrix SSDs: Prisoner's Dilemma and Stag Hunt. Finally, we prove its success at defending against adversaries whilst maintaining cooperation in RL agents in Section \ref{sec:RL}. 
%%%%%%%%%%%%%%%%%%%%%%%%%%%%%%%%%%%%%%%%%%%%%%%%%%%%%%%%%%%%%%%%%%%%%%%%
\section{Mathematical Preliminaries} \label{sec:prelims}
   
 A game is a tuple $(N, A, u)$ where: $N=\{1,2,...,n\}$ is the set of agents; $A = A_1 \times ... \times A_n$ is the action space of all players, $n\geq 2$; and, utility functions $u = (u_1,...,u_n)$ where $u_i:A \rightarrow \mathds{R}$ is a convex utility function for player $i$. 
 The expected utility of player $i$ is denoted $E[u_i(\sigma_i, \sigma_{-i})]$.
 
 Denote the mixed strategy space of player $i$ as $\Sigma_i$ and opponent strategy space as $\Sigma_{-i} =\Pi_{j \neq i} \Sigma_j$.  Then the \textit{minimax value} of the game $v_i$ is the highest value player $i$ can guarantee without knowing their opponents' actions:  $v_i = \max_{\sigma_i \in \Sigma_i}\min_{\sigma_{-i} \in \Sigma_{-i}}  E[ u_i(\sigma_i,\sigma_{-i})]$.
 
 A strategy $\sigma_i$ is \textit{safe} if it guarantees at least the minimax value on average: $E[u_i(\sigma_i, \sigma_{-i})] \geq v_i$ for any $\sigma_{-i} \in \Sigma_{-i}$. A strategy $\sigma_i$ is \textit{$\epsilon$-safe} if, and only if, for all $\sigma_{-i} \in \Sigma_{-i}$ we have: $v_i - \epsilon \leq E[u_i(\sigma_i, \sigma_{-i})]$. The \textit{best-response} to a strategy $\sigma_{-i}$ is $BR(\sigma_{-i}):=\arg \max_{\sigma_i} E[u_i(\sigma_i,\sigma_{-i})]$. A Nash equilibrium is a strategy profile $\sigma$ such that $\sigma_i \in BR(\sigma_{-i})$ for all $i \in N$. %% here
 The Risk What You've Won in Expectation algorithm \cite{Ganzfried2012} plays an $\epsilon$-safe best response to a model of an opponent's strategy and achieves safety.
% The Risk What You've Won in Expectation algorithm plays an $\epsilon$-safe best response to a model of an opponent's strategy $M$ and achieves safety \cite{Ganzfried2012}. Pseudocode is shown in Algorithm \ref{algo:RWYWE} where the set of $\epsilon$-safe strategies is denoted $SAFE(\epsilon)$. Our proposed algorithm follows the same structure, thus, achieves safety.
%  \vspace{-0.3cm}
%      \begin{algorithm}%[h]
%     \SetAlgoLined
%      Initialize
%      $\epsilon_0 \gets 0$,
%      $v_i \gets $minimax value\;
%      \For{timestep $t=1$ to T}{
%           $\pi_i \gets \arg \max_{\sigma_i \in SAFE(\epsilon_t)} E[u_i(\sigma_i, M)]$\;
%           $i$ plays $\sigma_i$ from $\pi_i$\;
%           $-i$ plays $\sigma_{-i}$ from unknown $\pi_{-i}$\;
%           $\epsilon_{t+1} \gets \epsilon_t + E[u_i(\pi_i,\sigma_{-i})] - v_i$
%          }
%      \caption{Risk What You've Won in Expectation \cite{Ganzfried2012}} 
%      \label{algo:RWYWE}
%     \end{algorithm}
%  %\vspace{-0.3cm}
 
 Social dilemmas are games in which agents can either cooperate or defect where joint cooperation gains the highest total rewards but there is an incentive to deviate from this.  We will explore the conditions for beliefs in a sequential social dilemma  (SSD) to incentivize cooperation. The possible payoffs of each stage games are reward $R$ for mutual cooperation, punishment $P$ for mutual defection, sucker $S$ if exploited, and temptation $T$ for exploiting. Table~\ref{table:ssd} shows the form of a two-player matrix game to be social dilemma, such as Prisoner's Dilemma and Stag Hunt.

 Reinforcement learning (RL) problems are those which comprise an agent learning how to behave in an environment where they receive rewards for their actions. The environment is represented by a state variable, $s \in S$, and the principle task of the agent is to estimate the value of choosing an action, $a \in A$, given the current state. The goal is to then find an optimal policy stating the action to be taken in a given state to achieve the highest rewards. In multi-agent RL, multiple agents act in the same environment. 

 A \textit{Markov game} is a tuple $(N,S,(A^i),P,(R^i), \gamma)$ for $i \in N$: set of agents $N$ where $|N|>1$, state space $S$, action spaces $A^i$ with joint action space $A:= A^1 \times ... \times A^{|N|}$, Markovian transition model $P: S \times A \times S \mapsto [0,1]$, reward function $R^i: S \times A \times S \mapsto \mathds{R}$ and a discount factor $\gamma \in [0,1)$. The goal of an agent $i$ is to select a sequence of actions, or policy $\pi_i$, to maximize the cumulative discounted return $R_t= \sum_{k=0}^\infty \gamma^k r_{t+k+1}$. Such a sequence of actions is called the \textit{optimal policy} $\pi_i^*$. The joint policy $\pi: S \rightarrow \Delta(A)$ is defined as $\pi(a|s):= \prod_{i \in N} \pi_i(a_i|s)$. The state-value function for agent $i$, or \textit{value function}, $V_\pi:S \rightarrow \mathds{R}$ describes the expected value of following policy $\pi_i$ when your opponent follows policy $\pi_j$ from state $s$: $V^i_{\pi_i, \pi_{-i}}(s)= E[\sum_{t\geq 0} \gamma^t R^i_t| a^i_t \sim \pi^i(\cdot | s_t), s_0=s]$, where $-i$ represents all other agents in $N$.
 
%%%%%%%%%%%%%%%%%%%%%%%%%%%%%%%%%%%%%%%%%%%%%%%%%%%%%%%%%%%%%%%%%%%%%%%%
\section{Safe beliefs} \label{sec:safebeliefs}
 As a mathematical convenience to represent both our safety criteria and our beliefs about the distribution of opponent strategies, we introduce policy-conditioned beliefs. A \textit{policy-conditioned belief} is a function $p_i:\Sigma_i \rightarrow \Sigma_{-i}$.  Note that these beliefs are only a mathematical convenience and should not be thought of as a literal reflection of our setting, since players move simultaneously, the strategy of player $i$ does not causally affect the other players' strategies on the same time step. However, policy-conditioned beliefs allow us to represent both the idea that our opponent is drawn from some fixed distribution and the idea that our opponent might be adversarial, in a single formalism.  The set of \textit{best responses} to a policy-conditioned belief $p_i$ is defined as \[ BR(p_i) := \arg \max_{\sigma_i \in \Sigma_i}  E[u_i(\sigma_i, p_i(\sigma_i))]. \]

 %For some $\epsilon>0$, an $\epsilon$-Nash equilibrium is a strategy profile $\sigma$ and beliefs $p_i$ such that $u_i(\sigma_i,p_i(BR(p_i)) \leq u_i(BR(p_i),p_i(BR(p_i)) + \epsilon$ for all $i \in N$.
 A policy-conditioned belief is \textit{$\epsilon$-safe} if, and only if, for all $\sigma_{i} \in BR(p_i)$, $\sigma_i$ is $\epsilon$-safe. An example of a safe belief is the adversarial policy-conditioned belief, we define this policy-conditioned belief as \[ p_i^{A}(\sigma_i) := \arg \min_{\sigma_{-i} \in \Sigma_{-i}}E[u_i(\sigma_i,\sigma_{-i})]. \] 
 
 The safety property is desirable to minimize risk in environments where an adversarial opponent can exploit a policy. However, playing a safe strategy may end up with suboptimal outcomes, such as in Prisoner's Dilemma and Stag Hunt. Thus, we shall classify $\epsilon$-safe strategies that can safely cooperate with other agents in any multi-agent games. The payoff matrix for players is assumed to be normalized on $[0,1]$ for the purpose of simplicity; otherwise, there is an additional coefficient of the range of payoffs\footnote{For bounded utilities where the greatest range in payoff a player can have is $K$, replace $\epsilon$-safe with $K\epsilon$-safe in Propositions \ref{prop:closetoNEissafe} and \ref{prop:safebeliefs}.}. 

    \begin{prop}\label{prop:closetoNEissafe}
    For any two-player game with a Nash equilibrium $\sigma^*$, $\forall \epsilon \in [0,1]$ and $\forall i \in N$, $\exists \sigma_i \in \Sigma_i$ such that $||\sigma_i - \sigma^*_i||_\infty \leq \epsilon$ and $\sigma_i$ is $\epsilon$-safe\footnote{For $\bm{x}=(x_1,..,x_n)$, the supremum norm is $||\bm{x}||_\infty:= \sup \{|x_1|, ..., |x_n| \}$.}.
    \end{prop}
    \begin{proof}
    Since we assumed the payoffs are normalized, we have \[ \max_{\sigma_i} E[u_i(\sigma_i,BR( \sigma_i))]= 1 \text{  and  } \min_{\sigma_i} E[u_i(\sigma_i,BR( \sigma_i))]= 0.\] Let $\sigma^{min}_i$ be such that $E[u_i(\sigma^{min}_i,BR( \sigma^{min}_i))]= 0$.
    For any $\sigma_i$, such that $||\sigma_i - \sigma^*_i||_\infty \leq \epsilon$, we can bound the expected utility loss by the worst case:
    \begin{align*}
    & E[ u_i(\sigma^*_i,BR( \sigma^*_i))] -  E[u_i(\sigma_i,BR( \sigma_i))] \leq \, \, \, \, \, \, \,\, \, \, \, \, \, \, \\ & \, \, \, \, \, \, \, \, \, \, \, \,\, \, \, \,\, \, \, \, \, \, \,  E[ u_i(\sigma^*_i,BR( \sigma^*_i))] - E[u_i(\sigma^\epsilon_i,BR( \sigma^\epsilon_i))] ,
    \end{align*}
    where $\sigma^\epsilon_i = (1-\epsilon)\sigma^*_i + \epsilon \sigma^{min}_i$.
    We can bound the right hand side of the inequality using $ E[u_i(\sigma^*_i,BR( \sigma^*_i))] \leq 1$ and, since we have assumed convex utility functions, $E[u_i(\sigma^\epsilon_i,BR( \sigma^\epsilon_i))] \geq 1 - \epsilon$.
    Thus, we have $\epsilon$-safety as $E[u_i(\sigma^*_i,BR( \sigma^*_i))] - E[u_i(\sigma_i,BR( \sigma_i))] \leq \epsilon$. 
    \end{proof}
 Now that we have found conditions for the existence of $\epsilon$-safe strategies, we will prove similar properties for policy-conditioned beliefs. Define a policy-conditioned belief $p_i$ to be \textit{$\epsilon$-close} to a policy-conditioned belief $p_i'$ if, and only if, \[\max_{\sigma_i \in \Sigma_i}||p_i(\sigma_i)-p_i'(\sigma_{i})||_\infty \leq \epsilon.\] 
    
    \begin{prop} \label{prop:closebeliefscloseutilities}
    For any belief $p_i$ that is $\epsilon$-close to $p^{A}_i$, $\forall \sigma_i \in \Sigma_i$, \[ E[u_i(\sigma_i, p_i(\sigma_{i}))] - E[u_i(\sigma_i, p^{A}_i(\sigma_{i}))] \leq \epsilon. \]
    \end{prop}
    \begin{proof}
    For all $\sigma_{i} \in \Sigma_i$ and any policy-conditioned belief $p_i$ such that $p_i$ is $\epsilon$-close to $p^{A}_i$, can be bounded by
    \[ p_i(\sigma_{i}) = (1-\epsilon)p^{A}(\sigma_i) + \epsilon p^{max}_i(\sigma_i)  \]
    where $p^{max}_i(\sigma_i) = \arg \max_{\sigma_{-i}} E[u_i(\sigma_i, \sigma_{-i})]$.
    We can write the expected utility of belief $p_i$ as $ E[u_i(\sigma_i, p_i(\sigma_{i}))] = E[u_i(\sigma_i,(1-\epsilon)p^{A}_i(\sigma_i) + \epsilon p^{max}(\sigma_i))]$.
    Since the utility function is convex we can find an upper bound: 
    \begin{align*}
        & E[u_i(\sigma_i,(1-\epsilon)p^{A}_i(\sigma_i) + \epsilon p^{max}(\sigma_i))] \leq \, \, \, \, \, \, \, \, \, \, \, \, \, \, \\ & \, \, \, \, \, \, \, \, \, \, \,\, \, \, \, \, \, \, (1-\epsilon)E[u_i(\sigma_i,p^{A})] + \epsilon E[u_i(\sigma_i, p^{max}_i(\sigma_i))].
    \end{align*}
    Rearrange the difference in expected utility of $p_i$ and $p_i^{A}$:
    \begin{align*}
     & E[u_i(\sigma_i, p_i(\sigma_{i}))] -E[ u_i(\sigma_i, p_i^{A}(\sigma_{i}))] \\
     &\leq  (1-\epsilon)E[u_i(\sigma_i,p^{A})] + \epsilon E[u_i(\sigma_i, p^{max}_i(\sigma_i))] - E[u_i(\sigma_i, p_i^{A}(\sigma_{i}))] \\
     &\leq  \epsilon (E[u_i(\sigma_i, p^{max}_i(\sigma_i))] - E[u_i(\sigma_i, p_i^{A}(\sigma_{i}))]) \leq \epsilon(1-0) = \epsilon 
      \end{align*} Therefore, the inequality holds true.
      \end{proof}
 Thus, beliefs close to $p^A$ have similar expected utilities. Now we can address the safety property in the same context.
    \begin{prop} \label{prop:safebeliefs}
      If the policy-conditioned belief $p_i$ is $\epsilon$-close to the adversarial policy-conditioned belief $p_i^{A}$, and utilities are bounded on $[0,1]$ then $p_i$ is $\epsilon$-safe.
    \end{prop}
    \begin{proof}
       Since $p_i$ is $\epsilon$-close to the adversarial policy-conditioned belief $p_i^A$, then by Proposition \ref{prop:closebeliefscloseutilities}, $\forall \sigma_{i} \in \Sigma_i$ we have \[ E[u_i(\sigma_i, p_i(\sigma_{i}))] - E[u_i(\sigma_i, p_i^{A}(\sigma_{i}))] \leq \epsilon. \]  Take $\sigma_{i} \in BR(p_i)$, then 
       \begin{align*} 
       E[u_i(\sigma_i, p_i(\sigma_{i}))] & =  \max_{\sigma_i \in \Sigma_i} E[u_i(\sigma_i, p_i(\sigma_{i}))] \\
       & \geq  \max_{\sigma_i \in \Sigma_i} E[u_i(\sigma_i, p^{A}_i(\sigma_{i}))] \, = \,\, v_i. %\\
       %& = & v_i.
       \end{align*} 
       Thus, $v_i - \epsilon \leq  E[u_i(\sigma_i, p_i^{A}(\sigma_{i}))]$. Hence,  $\forall \sigma_i \in \Sigma_i$, $\sigma_i$ is $\epsilon$-safe.  This $p_i$ is $\epsilon$-safe.
    \end{proof}
    
    Proposition \ref{prop:safebeliefs} shows that we can append any policy-conditioned belief with some level of an adversarial policy-conditioned belief and it will be $\epsilon$-safe for some $\epsilon$. Consequently, we can take a policy-conditioned belief that naively cooperates with any opponent and bound its safety through creating an uncertainty; either this belief is true or they face an adversarial opponent.
 
%%%%%%%%%%%%%%%%%%%%%%%%%%%%%%%%%%%%%%%%%%%%%%%%%%%%%%%%%%%%%%%%%%%%%%%%
\section{Cooperation Inducing Beliefs} \label{sec:coopbeliefs}
 In the previous section, we proved that we could take a cooperation inducing policy-conditioned belief and still maintain the safety property through uncertainty about whether the opponent faced is, in fact, an adversary. Now we must develop a policy-conditioned belief that we know will incentivize cooperative behavior in sequential social dilemmas. First, we consider two-player matrix SSD games. % maybe references to reciprocity, translucent gt and amTFT here?
 
 Let the strategy of player $i$ be $(\alpha, \bar{\alpha})$ where $\alpha \in [0,1]$ is the intended probability of cooperating in the next round and $\bar{\alpha} \in [0,1]$ is the probability of cooperating for all subsequent rounds. This format is chosen for simplicity. Similarly, the current and future strategies of their opponent are denoted by $(\beta, 1-\beta, \bar{\beta}, 1-\bar{\beta})$. For discounted future returns, define the \textit{expected returns} $V_i: \Sigma_i \times \Sigma_i \times \Sigma_{-i} \times \Sigma_{-i}  \rightarrow \mathds{R}$ as
 \begin{align*}
     V_i((\alpha,\bar{\alpha}),(\beta,\bar{\beta})):= \,  E[u_i(\alpha, \beta)] \, + \, \,\, \sum_{t=1}^{n-1} \gamma^t E[u_i(\bar{\alpha},\bar{\beta})] 
 \end{align*}
 where $\gamma \in (0,1]$ is the discount factor. A policy-conditioned belief in the sequential game is a function $p_i:\Sigma_i \times \Sigma_i \rightarrow \Sigma_{-i} \times \Sigma_{-i}$.

 Although opponent strategies are unseen, suppose that a player believes that their chosen strategy for the next round will change the strategy of their opponent for all subsequent rounds depending on their level of cooperation. For some $x \in (0,1]$, if player $i$ chooses to cooperate with at least proportion $x$ then their opponent's future cooperation level will not decrease, and for cooperation less than $x$, their opponent's level of cooperation will not increase for future rounds. Let player $i$ such a policy-conditioned belief $p^C_i$, where $C$ stands for cooperation-promoting, formally defined as 
    \[  p^C_i(\alpha):= 
    \begin{cases}
       (\beta, \beta^+) & \alpha \geq x \\
        (\beta, \beta^-)              & \alpha < x 
    \end{cases}
    \]
 for some $x \in (0,1]$ where $\beta \leq \beta^+$ and $\beta \geq \beta^-$. This belief is similar to a belief the opponent plays tit-for-tat with the additional condition that the opponent can view intended mixed strategies and is willing to invest or reject risk capital in the long-term. 
 
 Let us find the necessary conditions on $p^C_i$ for cooperation to be an equilibrium strategy against similar agents, and hence, be a cooperation inducing policy-conditioned belief as required. For cooperation to occur, we require that defection is not a best-response, i.e. $\forall \alpha > 0$, $V_i((\alpha,\bar{\alpha}),p^C_i(\alpha)) \geq V_i((0,\bar{\alpha}),p^C_i(0,1))$.
    \begin{prop} \label{eq:coopbeliefs}
     For any matrix SSD and policy-conditioned belief $p^C_i$ where $\beta, \beta^+, \beta^-$ satisfy
    \[
        \alpha \beta(R+P-S-T)+\alpha(S-P)+ \sum_{t=1}^{n-1} \gamma^t (\beta^+-\beta^-)[\bar{\alpha}(R+P-S-T)+T-P] \geq 0   
    \]
    for some $\bar{\alpha} \in [0,1]$, has cooperation as a best-response. 
    \end{prop}
    \begin{proof}
    For cooperation to be a best-response, we must have $V_i((\alpha,\bar{\alpha}),p^C_i(\alpha)) \geq V_i((0,\bar{\alpha}),p^C_i(0,1))$.
       %% *** sort formatting here ***!!!!!
       We can write the expected returns of cooperating with positive probability in terms of payoffs as 
       \begin{align*}  V_i((&\alpha,\bar{\alpha}),p^C_i(\alpha)) =  \alpha \beta R  + \alpha(1-\beta)S  + \beta(1-\alpha)T + (1-\alpha)(1-\beta)P \\ & +  \sum_{t=1}^{n-1}\gamma^t[\bar{\alpha} \beta^+R \, + \bar{\alpha}(1-\beta^+)S + \beta^+(1-\bar{\alpha})T + (1-\bar{\alpha})(1-\beta^+)P]. 
       \end{align*}
       The expected returns of defecting ($\alpha = 0$) is
        \begin{align*}  
        V_i((&0,\bar{\alpha}),p^C_i(0,1)) = \beta T \, +  (1-\beta)P + \sum_{t=1}^{n-1} \gamma^t [\bar{\alpha} \beta^-R +  \bar{\alpha}(1-\beta^-)S + \\ & \beta^-(1-\bar{\alpha})T + (1-\bar{\alpha})(1-\beta^-)P].
       \end{align*}
       Now by substituting these into the inequality we get
        \[
        \alpha \beta(R+P-S-T)+\alpha(S-P)+ \sum_{t=1}^{n-1} \gamma^t (\beta^+-\beta^-)[\bar{\alpha}(R+P-S-T)+T-P] \geq 0,   
        \]
        as given. 
       %% old one:
    %   This is equivalent to
    %     \begin{align*}
    %      & \alpha \beta R  + \alpha(1-\beta)S  + \beta(1-\alpha)T + (1-\alpha)(1-\beta)P+  \sum_{t=1}^{n-1}\gamma^t[\bar{\alpha} \beta^+R \, + \\ 
    %      & \bar{\alpha}(1-\beta^+)S + \beta^+(1-\bar{\alpha})T + (1-\bar{\alpha})(1-\beta^+)P]  \geq \beta T \, +  (1-\beta)P + \\ 
    %     & \,\,\,\,\,\,\,\,\,\,\,\, \sum_{t=1}^{n-1} \gamma^t [\bar{\alpha} \beta^-R +  \bar{\alpha}(1-\beta^-)S + \beta^-(1-\bar{\alpha})T + (1-\bar{\alpha})(1-\beta^-)P] 
    %     \end{align*}
    %     A simple rearrangement shows equivalence with the given inequality. 
    \end{proof}

 If $\beta(R+P-S-T)+(S-P)>0$, then $E[u_i(\alpha, 1-\alpha, \beta, 1-\beta)] $ is increasing in $\alpha$. As such, full cooperation will be dominant, so $i$ will play $\alpha=1$. Otherwise, they will play $\alpha=x$, in which case $x$ should be set to 1 to induce fully cooperative behavior. $\beta^+-\beta^-$ is the change in cooperation that the next strategy will induce in the opponent's future strategy.

 These beliefs can be summarized as those who believe their next action will affect opponents' future strategies such that they prefer to cooperate now to avoid a reduction in future utility. We call a strategy that follows this system $p^C$. Other beliefs from the literature that are cooperation promoting could also be substituted here, such as translucency \cite{Capraro2019}.
%%%%%%%%%%%%%%%%%%%%%%%%%%%%%%%%%%%%%%%%%%%%%%%%%%%%%%%%%%%%%%%%%%%%%%%%    
\section{Trade-Off Between Cooperation and Safety}\label{sec:coopvsafe}
In the previous sections, we described how to use policy-conditioned beliefs to promote both safety and cooperation.  In this section, we will use these ideas to show how the two concepts are necessarily in tension but how this tension disappears in the long-run.  To make this concrete, we will define the value of cooperation to be the value expected against the policy-conditioned belief $p^C_i$, that is:
\[
V^C(\pi_i) = E[u_i(\pi_i, p^C_i(\pi_{i}))].
\]

Similarly, we can think of the value of safety to be the value expected against the policy-conditioned belief $p^A_i$, that is:
\[
V^A(\pi_i) = E[u_i(\pi_i, p^A_i(\pi_{i}))].
\]

We will define $\overline{V}^C = \max_{\pi}\{V^C(\pi_i)\}$ to be the optimal cooperative value and $\overline{V}^A = \max_{\pi}\{V^A(\pi_i)\}$ to be the optimal safe value.  For the sake of simplicity, we will assume that $\beta = \beta^- = 0$ and $x=1$, which would be the case for cooperation-promoting policies when defection is a safe strategy and cooperation is socially optimal. 

When we focus on the decision made in the first round, the trade-off is fairly stark. The more likely you are to cooperate, the more likely you are to be exploited. However, this also increases the likelihood of encouraging your opponent to cooperate in the following round(s).  In the second round, the story is largely the same.  Since we assume $\beta = \beta^- = 0$, the cooperation-promoting policy would defect on the first step.  Thus, the safety that your policy loses over the first two steps is proportional to the probability that you cooperate, and the amount that you lose against the cooperation-promoting policy over the first two steps for defecting is proportional to the probability that you defect.

However, when we look at the third step, the trade-off is less severe.  If $\pi_i$ cooperated with proportion $\alpha$ on the first time step and $\alpha$ is high enough that the cooperation-promoting policy cooperates with them on the second time step, it increases the budget of safety for $\pi_i$, now being able to cooperate at $\alpha + E[u_i(\pi_i,\sigma_{-i})] - v_i$ without exceeding the safety budget. If the policy were any more cooperative, the adversary would be motivated to cooperate on the second round to be able to exploit in future rounds.

This pattern continues until it is safe for our policy to cooperate deterministically.  At which point, our policy will continue to cooperate for the rest of the round.  This means that the trade-off between safety and performance with cooperation-promoting policies can be characterized by the probability of not cooperating early on, which over the long run is a negligible fraction of overall performance.  This is captured formally in Proposition \ref{prop:tradeoff}.

\begin{prop} \label{prop:tradeoff}
    Suppose that defection is a safe strategy and cooperation is socially optimal.  Let $\pi_i$ be $\epsilon$-safe.  Let $\alpha_t$ be the probability $\pi_i$ cooperates on round $t$ against a cooperation-promoting belief and assume $E[r_t] \leq d \alpha_{t-1} + v_i$ for some constant $d>0$, as is the case when $\beta = \beta^- = 0$.  Then: \[\overline{V}^C  - V^C(\pi_i) \geq \frac{I}{T} \overline{V}^C - d\epsilon \frac{1-\Phi_C^{I+1}}{1-\Phi_C}.\]
    Where $C = \frac{d}{P-S}$, $\Phi_x = \frac{1 + \sqrt{1+4x}}{2}$, notated as such because $\Phi_1$ is the golden ratio, and $I= \min\{\lceil-\log_{\Phi_C}(\epsilon)\rceil,T\}$.  Moreover, when $E[r_t] = d \alpha_{t-1} + v_i$, there is an $\epsilon$-safe policy $\overline{\pi}_i$ which makes this bound tight.
\end{prop}
\begin{proof}
    Let $\pi_i$ be $\epsilon$-safe so $\epsilon = \overline{V}^A - V^A(\pi_i)$.  The result follows from the fact that cooperation at each round gives a bound for safety, which we define to be $\tilde{\alpha}_{k}$ as follows: 
    \[
    \alpha_k \leq \frac{\epsilon + \sum\limits_{t=0}^{k-1} E[r_t] - kv_i}{P-S} \leq \frac{\epsilon + \sum\limits_{t=0}^{k-1} (d \alpha_{t-1} + v_i) - kv_i}{P-S} =  \tilde{\alpha}_{k}.
    \] 
    
    If we take these $\tilde{\alpha}_i$ to inductively define a policy $\overline{\pi}_i$, then when $E[r_t] = d \alpha_{t-1} + v_i$, $\overline{\pi}_i$ can be proven achieve the desired equality by changing all of the inequalities in the following proof to equalities.  By substitution we have:
    \begin{align*}
    \tilde{\alpha}_{k} &= \frac{\epsilon + \sum\limits_{t=0}^{k-1} (d \alpha_{t-1} + v_i) - kv_i}{P-S} = \frac{\epsilon + d \sum\limits_{t=0}^{k-1}  \alpha_{t-1}}{P-S} 
    \\
    &= \frac{\epsilon + d \alpha_{k-2} + d \sum\limits_{t=0}^{k-2} \alpha_{t-1} }{P-S}
    = \frac{\epsilon + d \sum\limits_{t=0}^{k-2} \alpha_{t-1}}{P-S}  + \frac{d}{P-S}  \alpha_{k-2}
    \\
    &\leq \tilde{\alpha}_{k-1}  + C \tilde{\alpha}_{k-2} 
    \end{align*}

    %= (1+\frac{d}{P-S}) \tilde{\alpha}_{k-1} \leq \epsilon(1+\frac{d}{P-S})^{k-1}. 
    
    We can then show by induction that $\tilde{\alpha}_{t} \leq \epsilon \Phi_C^t$.  In the base case $\tilde{\alpha}_{0} = \epsilon = \epsilon \Phi_C^0$ and in the inductive case:
    \begin{align*}
    \tilde{\alpha}_{k} \leq \tilde{\alpha}_{k-1}  + C \tilde{\alpha}_{k-2} = \epsilon \Phi_C^{t-1} + C \epsilon \Phi_C^{t-2} = \epsilon (\Phi_C + C) \Phi_C^{t-2} = \epsilon\Phi_C^{t}.
    \end{align*}
    %     \begin{align}
    % \alpha_k &\leq \frac{\epsilon + \sum\limits_{t=0}^{k-1} (g * \alpha_{t-1} + v_i) - k*v_i}{P-S} \\
    % &= \frac{\epsilon + g * \sum\limits_{t=0}^{k-1}  \alpha_{t-1}}{P-S} 
    % \\
    % &= \frac{\epsilon + g * \sum\limits_{t=0}^{k-2} \alpha_{t-1} + g * \alpha_{k-2}}{P-S}
    % \\
    % &= \frac{\epsilon + g * \sum\limits_{t=0}^{k-2} \alpha_{t-1}}{P-S}  + \frac{g}{P-S} * \alpha_{k-2}
    % \\
    % &\leq \alpha_{k-2}  + \frac{g}{P-S} * \alpha_{k-2}
    % \\
    % &= (1+\frac{g}{P-S}) \alpha_{k-2}
    % \\
    % &= (1+\frac{g}{P-S})^{\lfloor \frac{k}{2}\rfloor} \epsilon.
    % \end{align}
    
Where the last inequality follows from the fact that $\Phi_C+C =\Phi_C^2$.  This inequality and $\alpha_k \leq 1$, allow us to bound $\overline{V}^C  - V^C(\pi_i)$ by simply summing the expected rewards:
\begin{align*}
    V^C(\pi_i) & = \sum\limits_{t=0}^T E[r_t] \leq d \sum\limits_{t=0}^T \alpha_{t-1} +  T v_i 
              \leq d \sum\limits_{t=0}^T \min\{\epsilon\Phi_C^{t}, 1\} +  \overline{V}^A.
\end{align*}
% \begin{align}
%     V^C(\pi_i) & = \sum\limits_{t=0}^T E[r_t] \\
%               & \leq g* \sum\limits_{t=0}^T \alpha_{t-1} +  T * v_i \\
%               & \leq g* \sum\limits_{t=0}^T \min\{\epsilon(1+\frac{g}{P-S})^{\lfloor \frac{k}{2}\rfloor} \epsilon, 1\} +  \overline{V}^A.
% \end{align}

Since the socially optimal behavior is deterministic cooperation, the difference from optimal behavior only occurs when $\epsilon\Phi_C^{t} < 1$ which happens when $t<-\log_{\Phi_C}(\epsilon)$. If $-\log_{\Phi_C}(\epsilon)$ is bigger than $T$, the second term of the min would never occur in the summation.  Thus, the point where the summation starts using the second term is $I= \min\{\lceil-\log_{\Phi_C}(\epsilon)\rceil,T\}$. We plug this into the final inequality and use the geometric series formula to get:
\begin{align*}
V^C(\pi_i) &\leq d \sum\limits_{t=0}^{I} \epsilon(\Phi_C)^{t} +  \overline{V}^C (1-\frac{I}{T}) 
 =  d\epsilon \frac{1-\Phi_C^{I+1}}{1-\Phi_C}+\overline{V}^C (1-\frac{I}{T})
\end{align*}

We multiply both sides by $-1$ and add $\overline{V}^C$ for the desired result.  

%Upper Bound Proof%  In the case when $E[r_t] = d \alpha_{t-1} + v_i$ the policy, $\overline{\pi}$ can be constructed by cooperating with probability $\tilde{a}_k$.  The lower bound on the performance of $\overline{\pi}$,  which can be seen by replacing the upper bound for the Fibonacci numbers with the appropriate lower bound $\tilde{a}_{k} \geq \epsilon \Phi_C^{t-1}$, and following through the proof using the same techniques.
\end{proof}

It is important to note that this trade-off, after a certain point, does not grow with $T$.  Moreover, the return on cooperation value for small reductions in the optimal safety value grows very quickly, as the small loss in safety can effectively be reinvested at each successive time step as the policy receives the gains from cooperation. Thus, in long iterated games, the optimal policy for this objective is nearly-optimal for the designer to deploy into either fully adversarial settings or fully cooperative settings. In the next section, we take this core insight as the motivation for an algorithm that uses the ideas of reinvesting this risk capital in order to achieve high degrees of both safety and cooperation.

\section{Accumulating Risk Capital Through Investing in Cooperation} \label{sec:arctic}
  In SSDs, if an opponent is cooperative, surplus payoff above the value of the game can be reinvested for higher long-term gain. However, we also need to be able to invest safely to avoid exploitation. 
  
  To begin with, we give agents an adversarial policy-conditioned belief to avoid exploitation and maintain safety\footnote{We could forgo perfect safety in the first round to increase cooperation between agents initially for better long-term rewards with only a small amount of risk.}. For any surplus capital gained, $\epsilon$, we can deviate our safe beliefs to be $\epsilon$-safe and begin to reinvest to build trust with opponents. One such belief is $p^{\epsilon C}_i := (1-\epsilon) p_i^{A}+\epsilon p^C_i$, where the sequential adversarial belief is \[ p_i^{A}(\sigma)=\arg \min_{\sigma_{-i} \in \Sigma_{-i}} E[u_i(\sigma_i,\sigma_{-i})],\] and $\beta, \beta^+, \beta^-$ satisfy the conditions in Proposition \ref{eq:coopbeliefs}. By Proposition \ref{prop:safebeliefs}, $p^{\epsilon C}_i$ is $\epsilon$-safe.

 Again, consider the two-player matrix SSDs of the form in Table \ref{table:ssd}. If $p_i^{A}(\sigma_i)=(0,1,0,1)$, which is true for both Stag Hunt and Prisoner's Dilemma, $\forall \sigma_i \in \Sigma_i$  we can write the belief $p^{\epsilon C}_i$ as
    \[  p^{\epsilon C}_i(\alpha)= 
    \begin{cases}
       (\epsilon \beta, \epsilon \beta^+) & \alpha \geq x \\
        (\epsilon \beta, \epsilon \beta^-)              & \alpha < x.
    \end{cases} \]
 The condition for cooperation to occur is now:
     \begin{equation} \label{eq:safecoopbeliefs}
        \alpha(P-S) - \epsilon \alpha \beta(R+P-S-T) \leq \sum_{t=1}^{n-1} \gamma^t \epsilon(\beta^+-\beta^-)[\bar{\alpha}(R+P-S-T)+T-P]  
    \end{equation}
    % \begin{equation} \label{eq:safecoopbeliefs}
    %     \epsilon \alpha \beta(R+P-S-T)+\alpha(S-P)+ \sum_{t=1}^{n-1} \gamma^t \epsilon(\beta^+-\beta^-)[\bar{\alpha}(R+P-S-T)+T-P] \geq 0  
    % \end{equation}
 Naturally, for $\epsilon=1$ this always holds. Ideally, we want to choose $\alpha, \beta, \beta^+, \beta^-$ such that this holds for the smallest possible $\epsilon$ so that we need to least amount of risk to allow cooperation to be a best response to these beliefs.
       \begin{figure}[t]
        \centering
        \includegraphics[width=\linewidth]{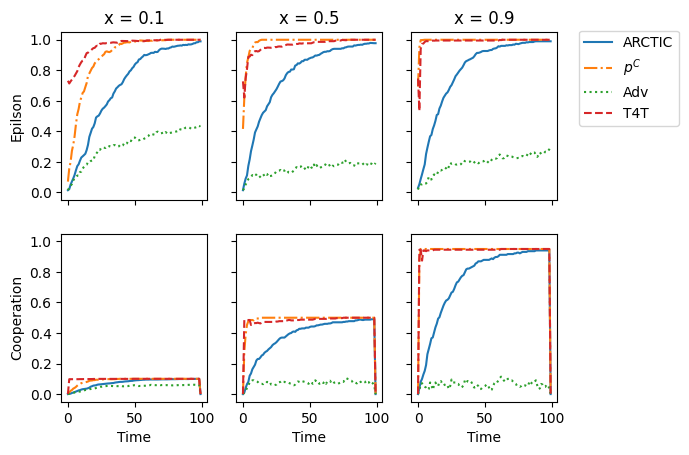}
        \caption{Simulations of ARCTIC playing different opponents in 100 rounds of Prisoner's Dilemma with various $x$. Against the adversarial strategy, ARCTIC does not learn to cooperate as there is not enough risk capital gained from their interactions. When two ARCTIC players interact, the risk capital slowly builds over time for all $x$ as they both play cautiously.}
        \label{fig:IPD_coop_ep}
        \Description{Graphs plotting cooperation and epsilon against time i.e., round number, for x values 0.1, 0.5, and 0.9. The cooperation for T4T and $p^C$ reaches x within the first few rounds for each x. With ARCTIC, x is reached at around round 90. Against the Adv player, cooperation stays below 0.1 for each x. The epsilon values increase over time for each value of x. With T4T and $p^C$ it reaches 1 at around time 40, 20 and 5 for values 0.1, 0.5, and 0.9 respectively. Against the ARCTIC opponent, epsilon reaches 1 at around 90-100 for each x. Against the Adv player, epsilon never exceeds 0.4.}
        \vspace{-10pt}
    \end{figure} 
 Pseudocode for the Accumulating Risk Capital Through Investing in Cooperation is given in Algorithm \ref{algo:ARCTIC}. Following from results in \cite{Ganzfried2012}, this algorithm is safe.
    \begin{algorithm}[ht]
    \SetAlgoLined
     Initialize
     $x \in [0,1]$,
     $\beta \in [0,1]$,
     $\beta^+ \gets 1$,
     $\beta^-  \gets 0$,
     $\epsilon \gets 0$,
     $v_i \gets $minimax value\;
     \For{$t=1$ to T}{
          $p_i \gets (1-\epsilon)p_i^{A}+\epsilon p_i^{C}$\;
          $\pi_i \gets \arg \max E[u_i(\sigma_i, p_i(\sigma_i))]$\;
          $i$ plays $\sigma_i$ from $\pi_i$, $-i$ plays $\sigma_{-i}$ from unknown $\pi_{-i}$\;
          $\epsilon \gets \min(\epsilon + E[u_i(\pi_i,\sigma_{-i})] - v_i, 1)$
         }
     \caption{ARCTIC}%Accumulating Risk Capital Through Investing in Cooperation}
     \label{algo:ARCTIC}
    \end{algorithm}

 We can think of $\epsilon$ as representing the amount of risk capital we are willing to invest in an opponent. For safe play in a sequential game, we begin with no risk capital and play the minimax strategy since we assume our opponent is adversarial. As the game goes on, the amount of risk capital will increase against non-adversarial opponents and we can safely invest such risk capital with the expectation of a return on the investment. Gradually, we build trust against similar opponents that reciprocate collaborative behaviors. 

 If our opponent has been cooperating in the past, then there is enough risk capital for ARCTIC to cooperate with such an opponent. However, if they then defect the amount of risk capital drops and they are more likely to defect in the next round, similar to how a tit-for-tat strategy punishes defections. If the risk capital is high enough, then punishment of a defection is less common since ARCTIC has learned to trust the good behavior of its opponent. This mechanism stops the strategy from being exploited against adversaries whilst maximizing cooperation with allies.
 
 \subsection{Matrix Games}
 To test the algorithm's performance, we simulated an ARCTIC agent for 100 rounds of Prisoner's Dilemma and Stag Hunt where its opponent either followed a simple strategy or best responded to their policy-conditioned beliefs. The simple strategies played against against were tit-for-tat (T4T) and pure defector (Adv), and the policy-conditioned belief opponents followed either ARCTIC or the cooperation inducing belief $p^C$. 
 Figures \ref{fig:IPD_coop_ep} and \ref{fig:SH_coop_ep} show the cooperation levels and risk capital $\epsilon$ for these experiments. Experiment parameters were: random action noise 5\%, $\bar{\alpha}=\alpha$, and the discount factor $\gamma=0.9$. Results were then averaged over 200 runs.% \footnote{For result reproducibility, see code at https://github.com/maunhb/matrix\_SSD\_games.}.
 
    \begin{figure}[t]
        \centering
        \includegraphics[width=0.95\linewidth]{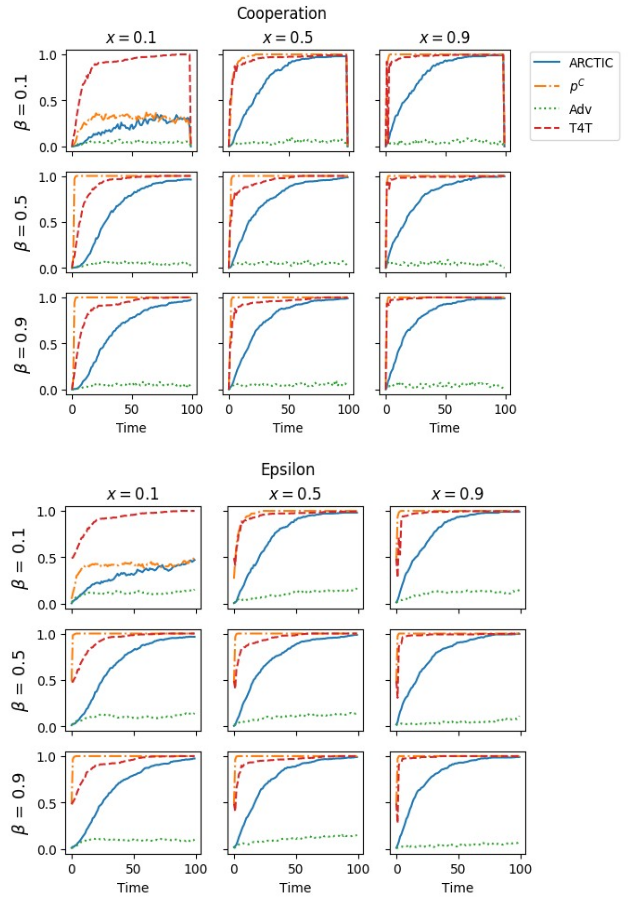}
        \caption{Simulations of 100 rounds of Stag Hunt whilst playing the ARCTIC strategy against 5 different opponent strategies. For cooperation to occur here, we need $x$ and $\beta$ to be large enough to satisfy condition \ref{eq:safecoopbeliefs}.}
        \label{fig:SH_coop_ep}
         \Description{Graphs plotting cooperation and epsilon against time i.e., round number, for x values 0.1, 0.5, and 0.9 and beta values of 0.1, 0.5 and 0.9. }
         \vspace{-12pt}
    \end{figure}
   In Prisoner's Dilemma (Figure \ref{fig:IPD_coop_ep}), the ARCTIC agent quickly learns to cooperate at rate $x$ against the cooperation incentivized T4T and $p^C$ players. When playing itself, the amount of risk capital, $\epsilon$, increases more gradually since they are learning to trust more cautiously than a T4T or $p^C$ player. Against the adversarial player, the cooperation levels a very low, and thus the ARCTIC agent maintains the safety property.
     \begin{figure*}[t]
     \centering
     \includegraphics[width=0.8\linewidth]{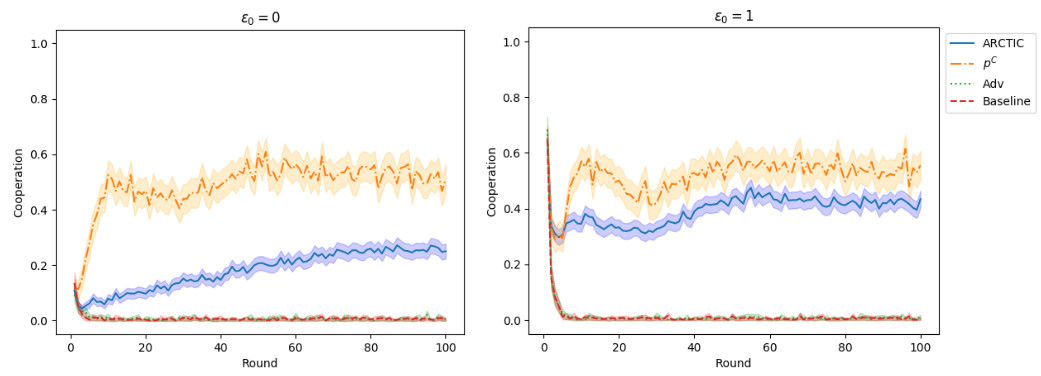}
     \caption{Cooperation of ARCTIC against different opponents over 100 rounds of Prisoner's Dilemma with initial risk capital of 0 (left) and 1 (right) when $x=0.5$. Against the adversarial and baseline opponents, ARCTIC learns to not cooperate, whereas when played against $p^C$ it cooperates around $x$. Against itself it cooperates at lower but increasing levels.}
     \Description{Graphs plotting cooperation against round number for $\epsilon_0$=0 and $\epsilon_1$=1. Against the $p^C$ player, the graphs look similar in that the cooperate at around 0.5 for most of the play, however, the initial cooperation is much lower for $\epsilon_0$=0. Against the ARCTIC player, with $\epsilon_0$=0, the cooperation slowly increase until about 0.2 at round 100. Whereas when $\epsilon_0$=1, it begins just below 0.4 and increases marginally. With the Adv and Baseline opponents with both values of $\epsilon_0$, the cooperation level quickly drops to 0.}
     \label{fig:pd_tournament}
      \vspace{-8pt}
 \end{figure*}

  Similarly, in Stag Hunt (Figure \ref{fig:SH_coop_ep}) the ARCTIC agent learns to cooperate with all but the adversary. For $x$ and $\beta$ large enough to satisfy condition (\ref{eq:safecoopbeliefs}), ARCTIC learns to fully cooperate by the end of the 100 rounds of play. Cooperation is much easier to achieve in Stag Hunt than in Prisoner's Dilemma since defect is not a dominant strategy. Against the adversary, not enough risk capital is collected for the best response to be cooperation which preserves the safety of the strategy. 
  
  Here, we have consider ARCTIC for payoffs normalized on $[0,1]$ as seen in Table \ref{table:ssd}. For games where payoffs have a range greater than 1, the $\epsilon$ update step should be \[ \epsilon \gets \min( \epsilon + \frac{1}{K} (E[u_i(\pi_i,\sigma_{-i})] -v_i), 1), \] where $K$ is the greatest difference between possible payoffs in order to normalize risk capital to be in $[0,1]$.
 
%%%%%%%%%%%%%%%%%%%%%%%%%%%%%%%%%%%%%%%%%%%%%%%%%%%%%%%%%%%%%%%%%%%%%%%%
\section{Multi-agent Reinforcement Learning} \label{sec:RL}
 SSDs are often subgames of more complex multiplayer games. These types of games are of particular interest in multi-agent RL due to the difficulty of learning cooperative policies. In more complex games where teamwork is required, it is only sequences of actions that create cooperative or selfish behaviors. Thus, we consider long-term behaviors in order to capture the nature of the agents. 
 
 We assume that the minimax value, $v$, of the game can be determined. An opponent can therefore detect whether the agent is cooperating by measuring whether their rewards are at least the value of the game. The level of cooperation is now
  $ x_t := \sum^t_{k=0} \gamma^{t-k} \mathds{1}_{r_k > v} $.
  
 Let player $i$'s cooperation inducing policy-conditioned belief $p^C_i$, be defined as 
 \[  p^C(\pi_i):= 
    \begin{cases}
       (\pi_j,\pi_j^+) & x^i_t \geq x \\
       (\pi_j,\pi_j^-) & \text{otherwise}
    \end{cases}
    \]
 where $V_{\pi_i^{p^C}, \pi_j^+}(s) > V_{\pi_i^{p^C}, \pi_j^-}(s)$ for some threshold $x \in (0,1]$. 
 
 If cooperation level $x^t_i$ is above a certain level, then the opponent will behave cooperatively. Otherwise, they will act in their own self-interest. To train agents with these cooperative beliefs, we can adapt the reward functions of their opponents as such:
     \[  r^{-i}_t \gets
    \begin{cases}
       r^i_t + r_t^{-i} & x^i_t \geq x \\
       r_t^{-i} & \text{otherwise}
    \end{cases}
    \]
    
 To train an agent with a policy-conditioned belief, it can be trained in an environment where those beliefs are true and transferred into the standard environment for deployment. 
 
 We trained distributed asynchronous advantage actor-critic (A3C) \cite{Mnih2016} agents on Prisoner's Dilemma and Stag Hunt environments. Agent policies were trained with the policy-conditioned beliefs $p^C$ and ARCTIC where $x=0.5$ and $\bar{\alpha}=\alpha$. Agents without beliefs were the Baseline agent (unmodified A3C) and the Adv agent. The neural network consists of two fully connected layers of size 32 and a Long Short Term Memory (LSTM) recurrent layer \cite{Gers1999}. The learning rate for Baseline, Adv, and $p^C$ agents was 0.001. For ARCTIC, the learning rates were 0.00007 and 0.0001 for Prisoner's Dilemma and
 Stag Hunt respectively. The entropy coefficient was 0.01. The state space for the ARCTIC agents was a onehot encoded $\epsilon$ value. Each agent was trained on 3 different random seeds and results are average across these policies for 300 rollouts.

 \subsection{Prisoner's Dilemma}
 The difficulty of playing Prisoner's Dilemma with a generic multi-agent RL algorithm is that defection is a strictly dominant strategy and, thus, usually converge to defecting. This means that a mechanism for agents to cooperate must be used to promote cooperation, which leaves them open to exploitation. By using ARCTIC here, the agent still acts optimally, but it acts optimally with respect to their policy-conditioned belief. 
 
   \begin{figure*}[ht!]
     \centering
     \includegraphics[width=0.8\linewidth]{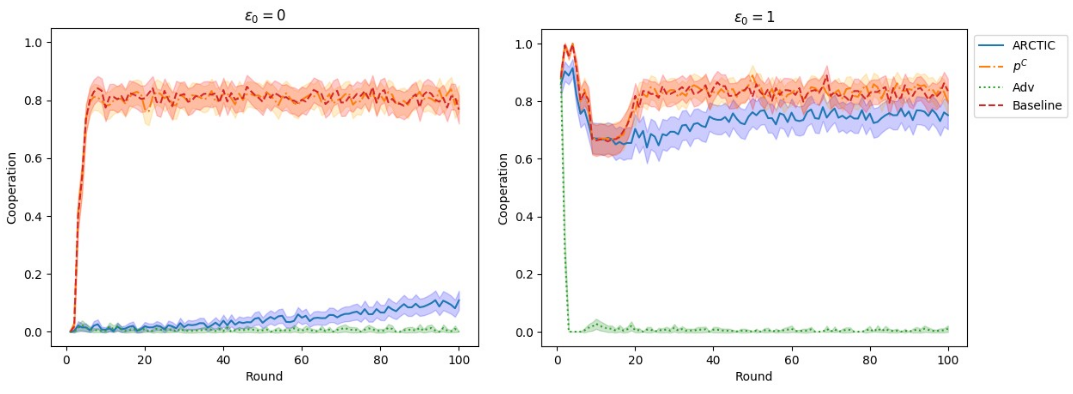}
     \caption{Cooperation of ARCTIC against different opponents over 100 rounds of Stag Hunt when starting with risk capital of 0 (left) or 1 (right). In both cases, ARCTIC cooperates the most with $p^C$ and baseline agents and the least with adversaries.}
     \label{fig:sh_tournament}
     \vspace{-5pt}
 \end{figure*}
 Table \ref{table:ipd_rl_results} shows the cumulative rewards for trained agents after 100 rounds of playing Prisoner's Dilemma. The baseline agent performs poorly against itself due to its inability to cooperate whereas the ARCTIC agent cooperates with itself some of the time and achieves a more socially optimal outcome. With the $p^C$ player, ARCTIC cooperates at higher levels than with itself. Against the adversary, ARCTIC achieves close to the value of the game on average. 
 \begin{table}[h]
     \caption{Agent scores for 100 rounds of Prisoner's Dilemma.}
     \label{table:ipd_rl_results}
     \vspace{-2mm}
     \begin{tabular}{c|c|c|c|c|}
     \multicolumn{1}{c}{} & \multicolumn{1}{c}{Baseline} & \multicolumn{1}{c}{ARCTIC} & \multicolumn{1}{c}{$p^C$} & \multicolumn{1}{c}{Adv} \\ \cline{2-5}
     Baseline & 25.01,\,25.01  & 25.54,\,24.83 & 70.29,\,9.91 & 25.00,\,25.01 \\ \cline{2-5}
     ARCTIC & 24.83,\,25.54  & 34.12,\,34.12 & 57.84,\,46.72 & 24.82,\,25.56  \\ \cline{2-5}
     $p^C$ & 9.91,\,70.29  & 46.72,\,57.84  & 55.21,\,55.21 & 9.89,\,70.34  \\ \cline{2-5}
     Adv & 25.01,\,25.00  & 25.56,\,24.82 & 70.34,\,9.89 & 25.00,\,25.00  \\ \cline{2-5}
    \end{tabular}
    \vspace{-2mm}
 \end{table}
 
 To try and improve the level of cooperation with itself, the initial level of risk capital can be increased to improve the cooperativeness with only increased risk in the first round of play. From Table~\ref{table:ipd_rl_results_eps_1}, we see that this leads to a more socially optimal outcome when played against itself without risking much against defectors. 
 \begin{table}[h]
   \caption{Scores for 100 rounds of Prisoner's Dilemma with $\epsilon_0 = 1$.}
     \label{table:ipd_rl_results_eps_1}
     \vspace{-2mm}
     \begin{tabular}{c|c|c|c|c|}
     \multicolumn{1}{c}{} & \multicolumn{1}{c}{Baseline} & \multicolumn{1}{c}{ARCTIC} & \multicolumn{1}{c}{$p^C$} & \multicolumn{1}{c}{Adv} \\ \cline{2-5}
     ARCTIC & 24.65,\,26.05 & 45.35,\,45.35 & 56.23,\,50.46 & 24.64,\,26.06  \\ \cline{2-5}
    \end{tabular}
    \vspace{-10pt}
 \end{table}

 \subsection{Stag Hunt}
 For an RL agent playing Stag Hunt, multi-agent RL algorithms learn to cooperate more effectively than Prisoner's Dilemma but this leaves them exploitable to adversaries. 
 
 In Table~\ref{table:sh_rl_results} the scores for agents in the Stag Hunt tournament can be found. Here, the baseline agent performs well against themself but achieve a poor score against adversaries. On the other hand, ARCTIC agents achieve the value of the game against adversaries.
  \begin{table}[h]
    \caption{Agent scores for 100 rounds of Stag Hunt.}
       \label{table:sh_rl_results}
        \vspace{-2mm}
     \begin{tabular}{c|c|c|c|c|}
     \multicolumn{1}{c}{} & \multicolumn{1}{c}{Baseline} & \multicolumn{1}{c}{ARCTIC} & \multicolumn{1}{c}{$p^C$} & \multicolumn{1}{c}{Adv} \\ \cline{2-5}
     Baseline & 99.53,\,99.53 & 78.61,\,94.31 & 99.78,\,99.34 & 0.13,\,74.84 \\ \cline{2-5}
     ARCTIC & 94.31,\,78.61  & 27.02,\,27.02  & 93.52,\,78.17 & 24.85,\,25.31  \\ \cline{2-5}
     $p^C$ & 99.34,\,99.78  & 78.17,\,93.52  & 98.66,\,98.66 & 0.34,\,74.43  \\ \cline{2-5}
     Adv & 74.84,\,0.13  & 25.31,\,24.85  & 74.43,\,0.34 & 25.02,\,25.02  \\ \cline{2-5}
    \end{tabular}
    \vspace{-2mm}
 \end{table}
 
  To encourage ARCTIC to cooperate more against itself, we can introduce a small amount of risk in the initial round. From Figure~\ref{fig:sh_tournament}, we see that ARCTIC achieves low levels of cooperation against themselves. %, but when $\epsilon_0=1$, they are willing to cooperate more. 
 When ARCTIC is equipped with a positive initial risk capital of $\epsilon_0 = 1$, the ARCTIC agent is able to achieve more socially optimal outcomes against all players except for the adversary where it achieves marginally less than when $\epsilon_0 = 0$. See results in Table~\ref{table:sh_rl_results_eps_1}.
  \begin{table}[h]
   \caption{Scores for 100 rounds of Stag Hunt with $\epsilon_0 = 1$.}
     \label{table:sh_rl_results_eps_1}
       \vspace{-2mm}
     \begin{tabular}{c|c|c|c|c|}
     \multicolumn{1}{c}{} & \multicolumn{1}{c}{Baseline} & \multicolumn{1}{c}{ARCTIC} & \multicolumn{1}{c}{$p^C$} & \multicolumn{1}{c}{Adv} \\ \cline{2-5}
     ARCTIC & 95.19,\,81.92 & 73.36,\,73.36 & 94.56,\,82.03 & 24.60,\,25.82  \\ \cline{2-5}
    \end{tabular}
      \vspace{-2mm}
 \end{table}
 
%%%%%%%%%%%%%%%%%%%%%%%%%%%%%%%%%%%%%%%%%%%%%%%%%%%%%%%%%%%%%%%%%%%%%%%%    
\section{Conclusion}

In summary, we studied the trade-off between cooperation and safety, first showing how to unify these two objectives in the formalism of policy-conditioned beliefs and then characterizing a trade-off between them.  We find that small risks to safety can lead to large returns in cooperation.  We made this trade-off more intuitive through the idea of risk capital and seeing cooperation as the compounding returns on its investment. We use this intuition to build Accumulating Risk Capital Through Investing in Cooperation (ARCTIC) which enacts this trade-off, achieving safe cooperation in iterated prisoner's dilemma and stag hunt.  

Cooperating while maintaining approximate safety allows us to design agents that individual developers would want to use out of their own self-interests. This is a promising development but leaves open questions that will be important in more complex environments. For instance, when there are many styles of successful cooperative strategies, agent designers would need to coordinate on a particular style of cooperation or build their agents to be adaptive to other agents' techniques. Moreover, although our method protects the agent against adversaries, it does not protect the agent against exploitative agents, who want to maximize their reward, which happens to come at the cost of our reward. Although an ARCTIC agent will not achieve less than $v_i- \epsilon$ in expectation, they could be exploited into accepting less than their fair share of the reward as long as they receive more than $v_i$.  This becomes more complex when combined with the coordination problems, as different coordination solutions could have different payouts which must be somehow distinguished from exploitative strategies. Extending the ideas of risk capital to these settings is left to future work.

There are also interesting challenges in scaling ARCTIC to larger environments. Our method is currently reliant on both knowing an expected minimax value and a clear notion of cooperation. In larger environments, these are both less accessible. To extend ARCTIC to these settings, either environment features would have to be estimated or the reliance on these features would have to be removed. Ultimately, addressing these issues could lead to general algorithms for safe multi-agent cooperation.%, which individual developers would find worth the risk.
 
%%% include description for those who cannot see eg:
%   \Description{Logo showing the words "AAMAS 2021, London, UK", 
%   with the 0 resembling the London Eye and the 1 resembling Big Ben.}
%% make sure figures are good in grayscale 
%%%%%%%%%%%%%%%%%%%%%%%%%%%%%%%%%%%%%%%%%%%%%%%%%%%%%%%%%%%%%%%%%%%%%%%%
\begin{acks}
  This work was supported by the Center for Human-Compatible AI and the Open Philanthropy Project.  We are grateful for funding of this work as a gift from the Berkeley Existential Risk Intuitive.
\end{acks}
%%%%%%%%%%%%%%%%%%%%%%%%%%%%%%%%%%%%%%%%%%%%%%%%%%%%%%%%%%%%%%%%%%%%%%%%

% To balance the columns on the final page of your paper, use the 
% `\texttt{balance}' package and issue the `\verb|\balance|' command
%  somewhere in the text of what would be the first column of the last 
%  page without balanced columns. This will be required for final papers.

%%% The next two lines define, first, the bibliography style to be 
%%% applied, and, second, the bibliography file to be used.

\bibliographystyle{ACM-Reference-Format} 
\bibliography{references}

%Make sure you provide complete and correct bibliographic information 
% for all your references, and list authors with their full names 
% (``Donald E.\ Knuth'') rather than just initials (``D.\ E.\ Knuth'').

%%%%%%%%%%%%%%%%%%%%%%%%%%%%%%%%%%%%%%%%%%%%%%%%%%%%%%%%%%%%%%%%%%%%%%%%

\end{document}